\documentclass[letterpaper,10pt,conference]{ieeeconf}  

\IEEEoverridecommandlockouts                              
\usepackage{mathptmx} 
\usepackage{times} 
\usepackage{amsmath} 
\usepackage{amssymb}  

\usepackage[noadjust,sort]{cite}
\usepackage{accents}
\newcommand{\ubar}[1]{\underaccent{\bar}{#1}}

\title{\LARGE \bf Katz Centrality of Markovian Temporal Networks: Analysis and Optimization}

\author{Masaki Ogura and Victor M.~Preciado
\thanks{The authors are with the Department of Electrical and Systems
Engineering, University of Pennsylvania, Philadelphia, PA 19014, USA.
Email:  {\tt\small \{ogura,preciado\}@seas.upenn.edu}}%
\thanks{This work was supported in part by the NSF under grants CNS-1302222 and IIS-1447470.}%
}

\usepackage{algorithm}
\usepackage{algorithmic}

\usepackage{color}

\newtheorem{definition}{Definition}[section]
\newtheorem{assumption}[definition]{Assumption}
\newtheorem{lemma}[definition]{Lemma}
\newtheorem{proposition}[definition]{Proposition}
\newtheorem{theorem}[definition]{Theorem}

\newtheorem{problem}[definition]{Problem}

\DeclareMathOperator*{\minimize}{minimize}
\DeclareMathOperator*{\subjectto}{subject\ to}

\DeclareMathOperator*{\col}{col}

\DeclareSymbolFont{bbold}{U}{bbold}{m}{n}
\DeclareSymbolFontAlphabet{\mathbbold}{bbold}
\newcommand{\onev}{\mathbbold{1}}

\usepackage{calrsfs}
\DeclareMathAlphabet{\pazocal}{OMS}{zplm}{m}{n}
\renewcommand{\mathcal}[1]{\pazocal{#1}}

\newenvironment{proofof}[1]{
\renewcommand\proof{\noindent\hspace{2em}{\itshape Proof of #1: }}%
\begin{proof}}{\end{proof}
}

\usepackage{mathtools}

\newcommand{\afterequation}{\vskip 3pt}

\begin{document}

\maketitle
\thispagestyle{empty}
\pagestyle{empty}

\begin{abstract}
Identifying important nodes in complex networks is a fundamental
problem in network analysis. Although a plethora of measures has been
proposed to identify important nodes in static (i.e., time-invariant)
networks, there is a lack of tools in the context of temporal networks
(i.e., networks whose connectivity dynamically changes over time). The
aim of this paper is to propose a system-theoretic approach for
identifying important nodes in temporal networks. In this direction,
we first propose a generalization of the popular Katz centrality
measure to the family of Markovian temporal networks using tools from
the theory of Markov jump linear systems. We then show that Katz
centrality in Markovian temporal networks can be efficiently computed
using linear programming. Finally, we propose a convex program for
optimizing the Katz centrality of a given node by tuning the weights
of the temporal network in a cost-efficient manner. Numerical
simulations illustrate the effectiveness of the obtained results.
\end{abstract}

\section{Introduction}

Identifying key nodes in complex networks is a fundamental problem in,
for example, social network analysis~\cite{Marsden2002}, viral
marketing~\cite{Hinz2011}, and biological networks~\cite{Girvan2002}.
In this direction, a variety of \emph{centrality measures} have been
proposed in the literature to assign importance scores to the nodes in
the network. For example, the PageRank~\cite{Gleich2014}, originally
introduced for ranking web pages~\cite{Page1998}, has found
application in a broad range of areas including chemistry, biology,
and neuroscience~\cite{Gleich2014}. Alternative centrality measures,
such as the Bonacich \cite{Bonacich1987}, Katz \cite{Katz1953}, or
HITS \cite{Kleinberg1999} centralities are also popular in the
analysis of complex networks. In practice, many complex networks of
practical interests present a \emph{time-varying} topology, as
frequently observed in human contact networks, online social networks,
biological, and ecological networks~\cite{Holme2015b}. In this
context, most of the centrality measures proposed for static
topologies are not able to faithfully capture the effect of temporal
variations on the importance of nodes~\cite{Braha2006}.

Although we find in the literature various generalizations of static
centrality measures for temporal networks, such as the
path-based~\cite{Lerman2010a}, betweenness~\cite{Tang2010a},
Katz~\cite{Grindrod2014a}, and PageRank~\cite{Rossi2012} centrality
measures, most of these generalizations are based on heuristic
arguments, without a rigorous mathematical justification. Furthermore,
there is also a lack of tools to optimize  the centrality measures in
the context of temporal networks (see
\cite{Fercoq2013,Csaji2014,Masson2015} for recent results on the
optimization of centrality measures of static networks). In this
paper, we extend the concept of Katz centrality~\cite{Katz1953} for
discrete-time temporal networks. Utilizing the theory of Markov jump
linear systems, we first show that the Katz centrality measure for
Markovian temporal networks is given as the solution of a linear
program. Based on this fact, we then propose an optimization framework
for increasing the centrality of a given node by tuning the weights of
the edges in the temporal network with a minimum cost.

This paper is organized as follows. After introducing necessary
mathematical notations, we define the Katz centrality for Markovian
temporal networks in Section~\ref{sec:Defn} and state the problem
studied in this paper. In Section~\ref{sec:analysis}, we present an
optimization framework for efficiently computing the Katz centrality.
We then give solutions to the problem of optimizing the Katz
centrality of Markovian temporal networks in Section~\ref{sec:opt}. We
illustrate the effectiveness of the proposed optimization approach in
Section~\ref{sec:numerical}.

\subsection{Mathematical Preliminaries}

For a positive integer $n$, define $[n] = \{1, \dotsc, n\}$. We let
$I_n$ and $O_n$ denote the $n\times n$ identity and the zero matrices,
respectively. By $\onev_p$, we denote the $p$-dimensional vectors
whose entries are all ones. A real matrix~$A$ (or a vector as its
special case) is said to be nonnegative, denoted by $A\geq 0$, if $A$
is entry-wise nonnegative. The notations $A\leq 0$, $A> 0$, and $A<0$
are understood in an obvious manner. For another matrix $B$, we write
$A\geq B$ if $A-B\geq 0$. Let $A$ be square. The spectral radius of
$A$ is denoted by $\rho(A)$. We say that $A$ is Hurwitz stable if the
eigenvalues of $A$ have negative real parts. We also say that $A$ is
Metzler if the off-diagonal entries of $A$ are all non-negative. The
Kronecker product~\cite{Horn1990} of two matrices $A$ and $B$ is
denoted by~$A\otimes B$. Given a collection of $n$ matrices
$A_1,\ldots, A_n$, we denote their direct sum by $\bigoplus_{i=1}^n
A_i$. If these matrices have the same number of columns, the matrix
obtained by stacking the matrices in vertical ($A_1$ on top) is
denoted by $\col(A_1, \dotsc, A_n)$.

A (weighted) graph is defined as a triple~$\mathcal G=(\mathcal V,
\mathcal E, w)$, where $\mathcal V = \{1, \dotsc, n\}$ is the set of
nodes, $\mathcal E \subset \mathcal V \times \mathcal V$ is the set of
edges  consisting of distinct and unordered pairs~$\{i, j\}$, and
$w\colon \mathcal E \to (0, \infty)\colon \{ i, j\}\mapsto w_{ij}$ is
the weights of edges. We say that a node $i$ is a neighbor of~$j$ (or
that $i$ and $j$ are adjacent) if $\{i, j\} \in \mathcal E$. The
adjacency matrix~$A\in \mathbb{R}^{n\times n}$ of the graph $\mathcal
G$ is defined as the matrix whose $(i,j)$-th entry is $w_{ij}$ if and
only if nodes $i$ and $j$ are adjacent, $0$ otherwise.
	
Finally, we recall basic facts about a class of optimization problems
called geometric programs~\cite{Boyd2007}. Let $x_1$, $\dotsc$, $x_m$
denote $m$ real positive variables. We say that a real-valued function
$f$ of~$x = (x_1, \dotsc, x_m)$ is a {\it monomial function} if there
exist $c>0$ and $a_1, \dotsc, a_m \in \mathbb{R}$ such that $f(x) = c
{\mathstrut x}_1^{a_1} \dotsm {\mathstrut x}_m^{a_m}$. Also, we say
that $f$ is a {\it posynomial function} if it is a sum of monomial
functions of~$x$. Given posynomial functions $f_0$, $\dotsc$, $f_p$
and monomial functions $g_1$, $\dotsc$, $g_q$, the optimization
problem
\begin{equation}\label{eq:generalGP}
\begin{aligned}
\minimize_x\ 
&
f_0(x)
\\
\subjectto\ 
&
f_i(x)\leq 1,\quad i=1, \dotsc, p, 
\\
&
g_j(x) = 1,\quad j=1, \dotsc, q, 
\end{aligned}
\end{equation}
is called a {\it geometric program}. It is known~\cite{Boyd2007} that
a geometric program can be easily converted into a convex optimization
problem.

\section{Katz Centrality for Temporal Networks}\label{sec:Defn}

In this section, we introduce the Katz centrality measure for
discrete-time temporal networks. We focus our attention to the
tractable case in which the process describing changes in the topology
of the network presents Markovian properties. We then state the
problem of optimizing the Katz centrality measure of a given node in a
Markovian temporal networks by tuning the weights of certain edges. 
Let us first introduce the class of temporal networks studied in this
paper. Let $L$ be a positive integer. For each $\ell \in [L]$, let
$\mathcal G_\ell$ be a weighted graph having nodes~$1, \dotsc, n$. We
call a discrete-time $\{\mathcal G_1, \dotsc, \mathcal G_L\}$-valued
stochastic process $\mathcal G = \{\mathcal G(k)\}_{k\geq 0}$ a
\emph{temporal network}. Each $\mathcal G_\ell$ is called a
\emph{layer} of the temporal network~$\mathcal G$. We say that
$\mathcal G$ is \emph{Markovian} if the stochastic process $\mathcal
G$ is a time-homogeneous Markov chain. We say that $\mathcal G$ is
\emph{i.i.d.} if the random variables $\mathcal G(k)$ ($k=0, 1,
\dotsc$) are independent and identically distributed.

We assume that a temporal network $\mathcal G$ is Markovian throughout
this paper. It is remarked that the class of Markovian temporal
networks includes several mathematical models of temporal networks,
including the edge-swapping model~\cite{Volz2009}, the activity-driven
model~\cite{Perra2012}, and the aggregated Markovian edge-independent
model~\cite{Ogura2015c}. We also note that the optimal intervention to
the spreading processes over continuous-time Markovian temporal
networks is studied in~\cite{Ogura2015a}. In order to motivate our
definition of Katz centrality for Markovian temporal networks, we here
recall the definition of the Katz centrality for static
networks~\cite{Katz1953}:

\begin{definition}[\cite{Katz1953}]\label{def:katz}
Let $\mathcal G$ be a weighted graph having $n$ nodes and
adjacency matrix~$A$. Let $\alpha < 1/\rho(A)$ be an arbitrary positive  parameter. Then, the
Katz centrality of $\mathcal G$ is defined as
\begin{equation}\label{eq:def:k}
v = (I-\alpha A)^{-1} \onev_n. 
\end{equation}
\afterequation
\end{definition}

Since this definition of the Katz centrality does not allow networks
to be time-varying, in this paper we utilize the following alternative
formulation of the Katz centrality. Let us consider the linear
time-invariant autonomous system
\begin{equation*}
x(k+1) = \alpha Ax(k),\ k\geq 0, 
\end{equation*}
with the initial condition $x(0) = \onev_n$. Since $x(k) = \alpha^k
A^k \onev_n$ for every $k\geq 0$, we can see that
\begin{equation*}
v = \sum_{k=0}^\infty (\alpha A)^k \onev_n = \sum_{k=0}^\infty x(k)
\end{equation*}
because $\alpha < 1/\rho(A)$ guarantees the convergence of the power
series. Based on this alternative expression, we can naturally
introduce the Katz centrality of Markovian temporal networks as
follows:

\begin{definition}\label{def:katz:markov}
Let $\mathcal G = \{\mathcal G(k)\}_{k\geq 0}$ be a Markovian temporal
network. Let $A(k)$ denote the adjacency matrix of the graph $\mathcal
G(k)$. Let $x$ be the solution of the discrete-time difference
equation
\begin{equation}\label{eq:myMJLS}
x(k+1) = \alpha A(k)x(k), \ x(0) = \onev_n, 
\end{equation}
where $\alpha$ is a positive constant. We define the \emph{Katz
centrality} of $\mathcal G$ as the vector 
\begin{equation}\label{eq:def:k(mathcalG)}
v = \sum_{k=0}^{\infty} E[x(k)]. 
\end{equation}
\afterequation
\end{definition}

We remark that the convergence of the power
series~\eqref{eq:def:k(mathcalG)} is not necessarily guaranteed for
all values of $\alpha$. We discuss the admissible range of $\alpha$ in
Section~\ref{sec:analysis}, where we also give an efficient method for
computing the Katz centrality using convex optimizations.

One of our main objectives in this paper is to optimize the Katz
centrality of a given node in a temporal networks by tuning edge
weights. As described below, we assume that tuning these weights has
an associated cost and our objective is to minimize the total tuning
cost. More formally, let $\mathcal G$ be a Markovian temporal network.
For each $\ell \in [L]$ and $\{i, j\} \in \mathcal E_\ell$, we let
$f_{\ell, ij}:[0, \infty) \to [0, \infty)$ be a function. For a
nonnegative scalar~$\Delta_{\ell, ij}$, the quantity~$f_{\ell,
ij}(\Delta_{\ell, ij})$ represents the cost for changing the weight of
the edge $\{i, j\}$ in $\mathcal G_\ell$ from $a_{\ell, ij}$ to
$a_{\ell, ij} + \Delta_{\ell, ij}$. If we define the
matrix~$\Delta_\ell = [\Delta_{\ell, ij}]_{i, j}$, then the sum
$\sum_{\{i, j\}\in\mathcal E_\ell} f_{\ell, ij}(\Delta_{\ell, ij})$
represents the cost for changing the adjacency matrix of the $\ell$th
layer from $A_\ell$ to $A_\ell + \Delta_\ell$. Let us denote the
weighted graphs having the resulting adjacency matrices by~$\mathcal
G'_1$, $\dotsc$,~$\mathcal G'_L$, and denote the resulting Markovian
temporal network by $\mathcal G'$. We notice that we do not consider
the design of the transition probabilities of the Markovian temporal
networks, which is indeed an important problem. We finally define
\begin{equation}\label{eq:def:C}
C(\Delta) =
\sum_{\ell=1}^L \sum_{\{i. j\} \in \mathcal E_\ell} f_{\ell,
ij}(\Delta_{\ell,ij}), 
\end{equation}
which represents the total cost for changing the temporal networks
from $\mathcal G$ to $\mathcal G'$.

Under the above notations, we consider the following optimization
problem in Section~\ref{sec:opt}:

\begin{problem}\label{prb:optimalCost}
Let $v'$ be the Katz centrality of the Markovian temporal
network~$\mathcal G'$. For each $\ell \in [L]$, let $\bar \Delta_\ell$
be an $n\times n$, symmetric, nonnegative matrix, representing the
maximum allowable change of the weights of the $\ell$th layer. Let $i
\in \{1, \dotsc, n\}$ and $\epsilon\geq 0$ be arbitrary. Find
symmetric and nonnegative matrices $\Delta_1$, $\dotsc$, $\Delta_\ell
\in \mathbb{R}^n\times n$ such that
\begin{equation}\label{eq:Deltaell.leq...}
\Delta_{\ell}\leq \bar \Delta_\ell
\end{equation}
for every $\ell \in [L]$, 
\begin{equation}\label{eq:rank constraints}
v'_i\geq (1+\epsilon)v'_j
\end{equation}
for every $j\neq i$, and $C(\Delta)$ is minimized.
\end{problem}

\section{Analysis} \label{sec:analysis}

The aim of this section is to present efficient methods for computing
the Katz centrality of Markovian temporal networks. Throughout this
paper, we let $P \in \mathbb{R}^{L\times L}$ denote the transition
probability matrix of the Markovian temporal network~$\mathcal G$. We
first prove the following proposition, which gives a closed-form
expression of the Katz centrality similar to \eqref{eq:def:k} for the
static case:

\begin{proposition}\label{prop:Katz}
Define the $(nL)\times (nL)$ matrix
\begin{equation}\label{eq:def:mathcalA}
\mathcal A = (P^\top \otimes I_n) \bigoplus_{\ell=1}^L A_\ell, 
\end{equation}
where $A_\ell$ is the adjacency matrix of the graph $\mathcal G_\ell$
for all $\ell \in [L]$. Then, the Katz centrality of $\mathcal G$ exists if
and only if $\alpha < 1/\rho(\mathcal A)$, under which we have
\begin{equation*}
v 
= 
(\onev_L^\top \otimes I_n) 
(I_{nL} - \alpha \mathcal A)^{-1} 
(\zeta_0 \otimes \onev_{n}), 
\end{equation*}
where the vector $\zeta_0\in \mathbb{R}^L$ is defined by 
\begin{equation*}
(\zeta_0)_\ell = P(\mathcal G(0) = G_\ell)
\end{equation*}
for each $\ell \in [L]$. 
\end{proposition}

For the proof of this theorem, we recall the following 
fact from the theory of Markov jump linear systems: 

\begin{lemma}[{\cite[Proposition~3.8]{Ogura2013f}}]
For each $\ell \in [L]$ and $k\geq 0$, define the $\{0, 1\}$-valued
random variable $\zeta(k)_\ell$ by
\begin{equation*}
\zeta(k)_\ell  = \begin{cases}
1, &\text{if }\mathcal G(k) = \mathcal G_\ell, 
\\
0, & \text{otherwise}. 
\end{cases}
\end{equation*}
Let $\zeta(k) = \col(\zeta(k)_1, \dotsc, \zeta(k)_L)$, and $x(k)$ be
the solution of the difference equation~\eqref{eq:myMJLS}. Then, it
holds that
\begin{equation}\label{eq:lifted}
E[\zeta(k+1)\otimes x(k+1)] = \alpha \mathcal AE[\zeta(k)\otimes x(k)]
\end{equation}
for every $k\geq 0$.
\end{lemma}

Let us prove Proposition~\ref{prop:Katz}: 

\begin{proofof}{Proposition~\ref{prop:Katz}}
From \eqref{eq:lifted}, we see that 
\begin{equation*}
\begin{aligned}
\sum_{k=0}^\infty E[\zeta(k)\otimes x(k)] 
&= 
\sum_{k=0}^\infty (\alpha \mathcal A)^k (\zeta_0 \otimes \onev_{n})
\\ 
&= 
(I_{nL}-\alpha \mathcal A)^{-1}(\zeta_0 \otimes \onev_{n}), 
\end{aligned}
\end{equation*}
where the convergence of the power series is guaranteed by the
assumption on $\alpha$ in Proposition~\ref{prop:Katz}. Since $\onev_L^\top \zeta(k) = 1$ for every
$k\geq 0$ with probability one, multiplying the matrix $\onev_L^\top
\otimes I_n$ from the left to this equation completes the proof of the
theorem.
\end{proofof}

Based on Proposition~\ref{prop:Katz}, we can further show that the
Katz centrality of Markovian temporal networks is given by the optimal
solution of a linear program, which allows an efficient computation of
the centrality for large-scale networks. For this purpose, we need to 
show the following preliminary result:

\begin{proposition}\label{prop:Katz analysis}
Assume that $\alpha < 1/\rho(\mathcal A)$. 
Let $\bar v \in \mathbb{R}^n$ be a positive vector. Then, we have
$v < \bar v$ if and only if there exist positive vectors
$\lambda_1, \dotsc, \lambda_L \in \mathbb{R}^n$ satisfying the
following inequalities:
\begin{align}
\zeta_k \onev_n + \alpha \sum_{\ell=1}^{L}p_{\ell k} A_\ell \lambda_\ell &< \lambda_k,\ \ell \in [L], \label{eq:smallInequ1} 
\\
\sum_{\ell=1}^L \lambda_\ell &< \bar v. \label{eq:smallInequ2}
\end{align}
\afterequation
\end{proposition}

\begin{proof}
Assume that $v < \bar v$. Notice that the matrix $\alpha \mathcal A -
I_{nL}$ is Metzler and, furthermore, Hurwitz stable because $\alpha
\in [0, 1/\lambda_{\max}(\mathcal A))$. We can, therefore,
take~\cite{Farina2000} a positive vector~$z \in \mathbb{R}^{nL}$ such
that $ (\alpha \mathcal A - I_{nL}) z< 0$. Since $v < \bar v$, we can take
$\epsilon>0$ such that
\begin{equation}\label{eq:k+eps<bar k}
v + \epsilon(\onev^\top_L \otimes I_n) z < \bar v.
\end{equation}
Define 
\begin{equation}\label{eq:def:lambda}
\lambda = (I_{nL}-\alpha \mathcal A)^{-1}(\zeta_0 \otimes \onev_n) + \epsilon z. 
\end{equation}
We can then show that $(I_{nL}-\alpha \mathcal A)\lambda = (\zeta_0
\otimes \onev_n) + \epsilon(I_{nL}-\alpha \mathcal A) z > \zeta_0
\otimes \onev_n$ and, therefore
\begin{equation}\label{eq:bigIneq1}
\zeta_0 \otimes \onev_n + \alpha \mathcal  A\lambda < \lambda.
\end{equation}
Also, from the definition~\eqref{eq:def:lambda} of $\lambda$, we see that 
\begin{equation}\label{eq:bigIneq2}
(\onev_L^\top \otimes I_n)\lambda 
=v + \epsilon (\onev_L^\top \otimes I_n) z < \bar v
\end{equation}
by \eqref{eq:k+eps<bar k}. Now, define positive vectors $\lambda_1$,
$\dotsc$, $\lambda_L \in \mathbb{R}^n$ by 
\begin{equation}\label{eq:def:lambdalambda}
\lambda = \col(\lambda_1,
\dotsc, \lambda_L)
\end{equation}
Then, it is straightforward to see that the inequalities~\eqref{eq:bigIneq1} and~\eqref{eq:bigIneq2} imply the inequalities~\eqref{eq:smallInequ1} and \eqref{eq:smallInequ2}, respectively.

On the other hand, assume that there exist positive
vectors~$\lambda_1$, $\dotsc$,~$\lambda_L \in \mathbb{R}^n$ satisfying
\eqref{eq:smallInequ1} and \eqref{eq:smallInequ2}. Define $\lambda$
by~\eqref{eq:def:lambdalambda}. We can see that $\lambda$ satisfies
\eqref{eq:bigIneq1} and \eqref{eq:bigIneq2}. From \eqref{eq:bigIneq1},
we have
\begin{equation}\label{eq:bigIneq1:follow}
\zeta_0 \otimes \onev_n < (I_{nL}-\alpha \mathcal A)\lambda.
\end{equation}Since
$\alpha \mathcal A - I_{nL}$ is Hurwitz stable, we have $(\alpha
\mathcal A - I_{nL})^{-1}\leq 0$ (see \cite{Farina2000}). Therefore,
multiplying $(\alpha \mathcal A - I)^{-1}$ to the both hand sides of
\eqref{eq:bigIneq1:follow}, we obtain  $(I_{nL}-\alpha \mathcal
A)^{-1}(\zeta_0 \otimes \onev_n) \leq \lambda$. In fact, since
$(\alpha \mathcal A - I_{nL})^{-1}$ does not have a zero-row and both
$\zeta_0\otimes \onev_n$ and $\lambda$ are positive, the strict
inequality~$(I_{nL}-\alpha \mathcal A)^{-1}(\zeta_0 \otimes \onev_n) <
\lambda$ holds. This inequality and \eqref{eq:bigIneq2} shows $v =
(\onev_L^\top \otimes I_n) (I_{nL}-\alpha \mathcal A)^{-1} (\zeta_0
\otimes \onev_n) \leq (\onev_L^\top \otimes I_n) \lambda < \bar v$, as
desired. This completes the proof of the proposition.
\end{proof}

We now provide a theorem that enables us to find
the Katz centrality of Markovian temporal networks by solving a linear
program:

\begin{theorem}\label{thm:KatzAsLP}
Assume that $\alpha < 1/\rho(\mathcal A)$. Then, the following linear program 
\begin{equation*}
\begin{aligned}
\minimize_{\lambda_\ell, \bar v 	 }\ \  & \sum_{\ell=1}^L \bar v_\ell
\\
\subjectto\ \  & \eqref{eq:smallInequ1} \text{ and } \eqref{eq:smallInequ2}
\end{aligned}
\end{equation*}
is solvable. Moreover, the optimal solution $\bar v^\star$ equals the
Katz centrality of $\mathcal G$.
\end{theorem}

The proof of this theorem is omitted because it is a direct
consequence of Proposition~\ref{prop:Katz analysis}.

\subsection{I.I.D.~Temporal Networks}

One of the drawbacks of the results presented so far is in the size of
the matrix~$\mathcal A$ (defined in \eqref{eq:def:mathcalA}), whose
size grows linearly with respect to both the number $n$ of the nodes
and the number $L$ of the layers in temporal networks. The latter
dependence is not desirable because $L$ can be significantly large
when the networks exhibit rather complicated dynamics. This subsection
shows that, in the case of i.i.d.~temporal networks, we can avoid the
possible computational complexity and rely on using a matrix whose
size equals always $n$, independent of~$L$.

Let $\mathcal G$ be an i.i.d.~temporal network taking values in the
set of weighted graphs~$\{\mathcal G_1, \dotsc, \mathcal G_L \}$
having nodes $1$, $\dotsc$, $n$. Let $p_\ell = P(\mathcal G(k) =
\mathcal G_\ell)$ for every $\ell \in [L]$. The next proposition
summarizes the computation of the Katz centrality of i.i.d.~temporal
networks:

\begin{proposition}\label{prop:iid}
Define the $n\times n$ matrix
\begin{equation*}
\mathcal A_{\text{iid}} = \sum_{\ell=1}^{L}p_\ell A_\ell, 
\end{equation*}
where $A_\ell$ is the adjacency matrix of $\mathcal G_\ell$ for each
$\ell \in [L]$. Then, the Katz centrality of $\mathcal G$ exists if
and only if $\alpha < 1/\rho(\mathcal A_{\text{iid}})$, under which we have
\begin{equation*}
v = (I_n - \alpha \mathcal A_{\text{iid}})^{-1} \onev_n. 
\end{equation*}
Furthermore, the Katz centrality is given as the solution of the linear program
\begin{equation*}
\begin{aligned}
\minimize_{\bar v}\ \  & \sum_{i=1}^N \bar v_i
\\
\subjectto\ \  & \onev_n + \alpha \mathcal A_{\text{iid}} \bar v < \bar v. 
\end{aligned}
\end{equation*}
\afterequation
\end{proposition}

\begin{proof}
Since $\mathcal G$ is i.i.d., we can easily see that $E[x(k+1)] =
\mathcal A_{\text{iid}}E[x(k)]$ for every $k\geq 0$. We therefore have
the formal power series $\bar v = \left(\sum_{k=0}^\infty \alpha^k
\mathcal A_{\text{iid}}^k \right) \onev_n$, which converges to $(I_n -
\alpha \mathcal A_{\text{iid}})^{-1} \onev_n$ if and only if $\alpha
<1/\rho(\mathcal A_{\text{iid}})$. We can prove the latter claim of
the proposition in the same way as the proofs of
Proposition~\ref{prop:Katz analysis} and Theorem~\ref{thm:KatzAsLP}.
\end{proof}

\section{Optimization}\label{sec:opt}

Using the analytical framework presented in the last section, this
section presents a convex optimization-based approach to the problem
of raising the Katz centrality of a given node, as stated in
Problem~\ref{prb:optimalCost}. In this paper, we place the
following reasonable assumption on the cost functions~$f_{\ell,ij}$,
stated as follows:

\begin{assumption}\label{asm:increasing}
The functions $f_{\ell, ij}$ are strictly increasing, continuous, and
satisfies $f_{\ell, ij}(0) = 0$ for all $\ell \in [L]$ and $\{i, j\}
\in \mathcal E_\ell$.
\end{assumption}

Then, instead of solving Problem~\ref{prb:optimalCost} directly, we
consider the following slightly modified problem, where we find the
optimal resource allocation, with a fixed budget, for demoting all the
nodes except the target node~$i$:

\begin{problem}\label{prb:}
Let $i \in \{1, \dotsc, n\}$, $\bar C>0$, and $\delta>0$ be arbitrary.
Find nonnegative, and symmetric matrices $\Delta_1$,
$\dotsc$, $\Delta_L \in \mathbb{R}^{n\times n}$ that minimizes $$(\max_{j\neq i}v'_j) + \delta
\sum_{j=1}^n v'_j$$ while satisfying  \eqref{eq:Deltaell.leq...}
and $C(\Delta) \leq \bar C$.
\end{problem}

We remark that the second term in the objective function is introduced
for the purpose of regularization. We can therefore take the
parameter~$\delta$ to be very small so that the main part of the
objective function effectively equals $\max_{j\neq i}v'_j$.

The next proposition shows that Problem~\ref{prb:} can be reduced to a
geometric program under a certain assumption on the cost functions:

\begin{theorem}\label{thm:preliminaryProblem}
If $C(\bar \Delta)\geq \bar C$, then the solution of
Problem~\ref{prb:} is given by the solution of the optimization
problem
\begin{subequations}\label{eq:KatzOptimization}
\begin{align}
\minimize_{\lambda_\ell,\,\bar v,\,\Delta_\ell}\ \ &(\max_{j\neq i}\bar v_j) + \delta \sum_{j=1}^n \bar v_j
\\
\subjectto \ \ 
& \eqref{eq:Deltaell.leq...} \text{ and }\eqref{eq:smallInequ2}, 
\\
&\zeta_k \onev_n + \alpha \sum_{\ell=1}^{L}p_{\ell k} (A_\ell+\Delta_\ell) \lambda_\ell < \lambda_k,\label{eq:mainConstraint}
\\
&C(\Delta)\geq \bar C.\label{eq:KatzOptimization:cost}
\end{align}
\end{subequations}
Moreover, if there exists a nonnegative and symmetric matrix~$\ubar
A_\ell \leq A_\ell$ for each $\ell \in [L]$ such that $-C(\Delta)$ is
posynomial in the entries of the matrices $\ubar A_\ell + \Delta_\ell$
($\ell=1, \dotsc, L$),  then the optimization problem
\eqref{eq:KatzOptimization} is a geometric program.
\end{theorem}

Before giving the proof for Theorem~\ref{thm:preliminaryProblem}, we
remark that the following straightforward formulation is not
appropriate in the current problem setting:
\begin{equation}\label{eq:KatzOptimization:doesn't work}
\begin{aligned}
\minimize_{\bar v,\,\lambda_\ell,\,\Delta_\ell}\ \ &(\max_{j\neq i}\bar v_j) + \delta \sum_{j=1}^n \bar v_j
\\
\subjectto \ \ 
&
\text{\eqref{eq:Deltaell.leq...}, \eqref{eq:smallInequ2}, and \eqref{eq:mainConstraint}}, 
\\
&C(\Delta)\leq \bar C.
\end{aligned}
\end{equation}
The only difference between this optimization problem and
\eqref{eq:KatzOptimization} is in the constraints for the cost
function. We can trivially see that the optimization
problem~\eqref{eq:KatzOptimization:doesn't work} yields the solutions
$\Delta^\star_1 = \cdots = \Delta^\star_L = 0$, which are meaningless
for solving Problem~\ref{prb:}.

We now give the proof of Theorem~\ref{thm:preliminaryProblem}

\begin{proofof}{Theorem~\ref{thm:preliminaryProblem}}
It is easy to see that the optimization
problem~\eqref{eq:KatzOptimization} is a geometric program under the
assumptions stated in the proposition. We shall prove the former claim
of the theorem. We remark that the existence of the optimal solution
is guaranteed by the assumption $C(\bar \Delta)\geq \bar C$.
Therefore, we need to show that the optimal solution, denoted by
$\Delta^\star$, satisfies $C(\Delta^\star) \leq \bar C$, as required
in Problem~\ref{prb:}.

Assume the contrary, i.e., $C( \Delta^\star) > \bar C$. Since
$C(\Delta^\star) > 0$ and $C$ is strictly increasing by
Assumption~\ref{asm:increasing}, there exists $\ell_0 \in [L]$ such
that $\Delta_{\ell_0} \neq O_n$. Define $\Delta'_\ell =
\Delta_\ell^\star$ for every $\ell \neq \ell_0$ and $\Delta'_{\ell_0}
= (1-\epsilon) \Delta_{\ell_0}^\star$ for a constant~$\epsilon>0$. By
the continuity of~$C$, there exists $\epsilon_0>0$ such that
$C(\Delta') \geq \bar C$. Define $\lambda'_\ell = \lambda^\star_\ell -
\epsilon_0\alpha p_{\ell_0
\ell}\Delta_{\ell_0}^\star\lambda^\star_\ell$ for each $\ell \in [L]$
and $\bar v' = \bar v^\star - \epsilon_0 \alpha \sum_{k=1}^L p_{\ell
k}\Delta^\star_{\ell_0} \lambda^\star_{\ell_0}$. Notice that, by
taking a sufficiently small $\epsilon_0$, we can guarantee that the
vectors $\lambda'_\ell$ and $\bar v'$ are positive. Then, we can
easily see that the triple $(\bar v', \lambda_\ell',
\Delta'_\ell)_{\ell \in [L]}$ satisfies the constraints in the
optimization problem~\eqref{eq:KatzOptimization}. Moreover, $\bar v'$
achieves a smaller value of the objective function than $\bar v^\star$
does. This however contradicts to the optimality of $\bar v^\star$.
\end{proofof}

Although Theorem~\ref{thm:preliminaryProblem} allows us to find the
optimal investment for suppressing the Katz centralities of all the
nodes except the target node~$i$ with a prescribed budget~$\bar C$,
the theorem does not directly allow us to find the minimum cost~$\bar
C$ for solving Problem~\ref{prb:optimalCost}. We therefore propose
Algorithm~\ref{alg:eigOptimization} where we use a binary search for
finding the minimum $\bar C$ achieving the
constraint~\eqref{eq:rank
constraints}. The effectiveness of Algorithm~\ref{alg:eigOptimization}
is illustrated in Section~\ref{sec:numerical} with numerical
simulations.

\begin{algorithm}[tb]                     
\caption{A procedure for solving Problem~\ref{prb:optimalCost}}         
\begin{algorithmic}                  
\REQUIRE $C_{\min}=0$, $C_{\max} = C(\bar \Delta)$, and $\delta C > 0$
\WHILE{$C_{\max}-C_{\min} > \delta C$}
\STATE $\bar C \leftarrow (C_{\max}+C_{\min})/2$
\STATE Solve the optimization problem~\eqref{eq:KatzOptimization}
\STATE $v \leftarrow v(\mathcal G')$
\IF{\eqref{eq:rank constraints} holds}
\STATE $C_{\max} = \bar C$
\ELSE
\STATE $C_{\min} = \bar C$
\ENDIF
\ENDWHILE
\end{algorithmic}
\label{alg:eigOptimization}
\end{algorithm}

\subsection{Decreasing Weights} 

This subsection briefly discusses the case where we can
\emph{decrease} the weights of edges by paying costs. We specifically
consider the following situation. Let $\mathcal G$ be a Markovian
temporal network. For each $\ell \in [L]$ and $\{i, j\} \in \mathcal
E_\ell$, we let $f_{\ell, ij}:[0, \infty) \to [0, \infty)$ be a
function satisfying Assumption~\ref{asm:increasing}. For a nonnegative
number~$\Delta_{\ell, ij}$, the quantity $f_{\ell, ij}(\Delta_{\ell,
ij})$ now represents the cost for \emph{decreasing} the weight of the
edge~$\{i, j\}$ in the layer~$\mathcal G_\ell$ from $a_{\ell, ij}$ to
$a_{\ell, ij} - \Delta_{\ell, ij}$. As is already done, we denote the
modified layers by $\mathcal G'_1$, $\dotsc$, $\mathcal G'_L$, and
denote the resulting Markovian temporal network by $\mathcal G'$. We
then define the total cost $C$ by \eqref{eq:def:C}.

In this dual problem setting, we can still prove the following theorem
that is an analogue of Theorem~\ref{thm:preliminaryProblem}:

\begin{theorem}\label{thm:preliminaryProblemDecreasing}
If $C(\bar \Delta) \geq \bar C$, then the solution of
Problem~\ref{prb:} is given by the solution of the optimization
problem~\eqref{eq:KatzOptimization:doesn't work}. Moreover, if there
exists a nonnegative and symmetric matrix $\ubar A_\ell \leq A_\ell$ for
each $\ell \in [L]$ such that $C(\Delta)$ is a posynomial in the
entries of $\Delta_\ell -\ubar A_\ell$ ($\ell=1, \dotsc, L$), then the
optimization problem~\eqref{eq:KatzOptimization:doesn't work} is a
geometric program.
\end{theorem}

The proof of this theorem is almost the same as the proof of
Theorem~\ref{thm:preliminaryProblem} and hence is omitted. Using
Theorem~\ref{thm:preliminaryProblemDecreasing}, we can construct an
algorithm, similar to Algorithm~\ref{alg:eigOptimization}, for solving
Problem~\ref{prb:optimalCost} in the case where we can decrease the
weights of edges. We omit the details due to limitations of the space.

\newcommand{\myfigwidth}{0.84\linewidth}

\begin{figure}[tb]
\centering
\includegraphics[width=\myfigwidth]{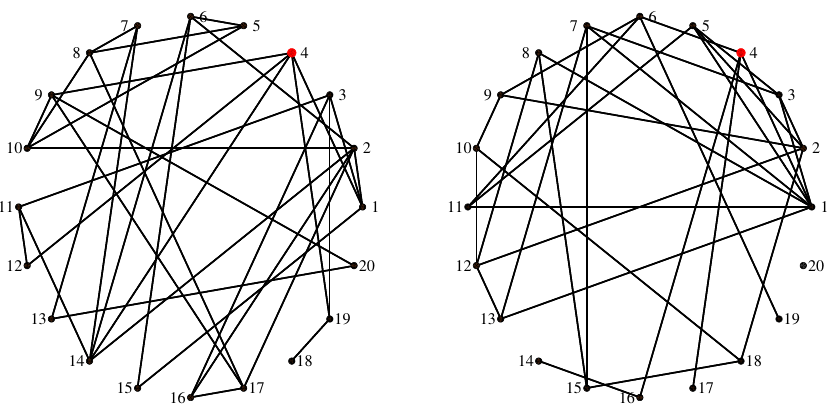}
\caption{Original graphs: $\mathcal G_1$ (left) and $\mathcal G_2$ (right)}
\label{fig:A1A2}
\vspace{.5 cm}
\centering
\includegraphics[width=\myfigwidth]{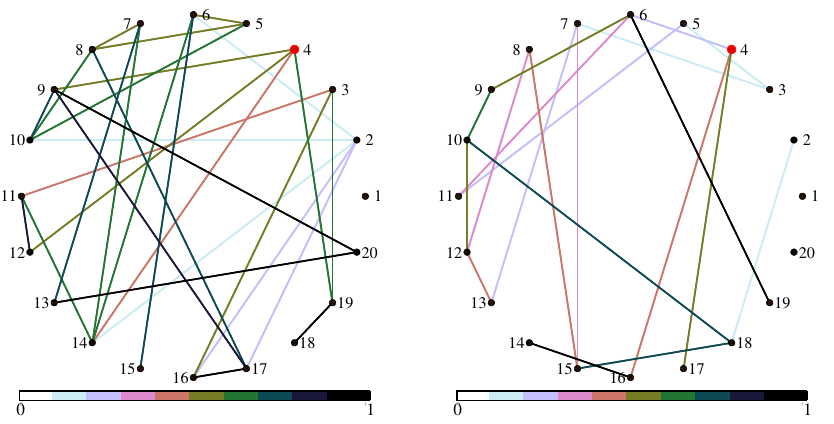}
\caption{The optimal additional weights: $\Delta_1^\star$ (left) and $\Delta_2^\star$ (right)}
\label{fig:Delta1Delta2}
\end{figure}

\begin{figure}[tb]
\centering
\includegraphics[width=\myfigwidth]{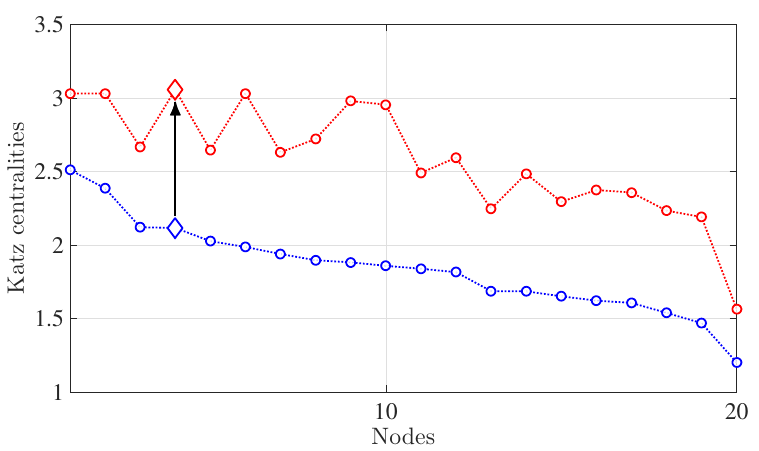}
\caption{Katz centralities: Blue: before optimization. Red: after optimization. Diamonds indicate the target node $i=4$.}
\label{fig:centralities}
\end{figure}

\section{Numerical Simulations}\label{sec:numerical}

We illustrate the results obtained in this paper with numerical
simulations. Let $\mathcal G_1$ and $\mathcal G_2$ be two independent
realizations of Erd\"os-R\'enyi graphs with $n=20$ nodes and edge
probability~$p=0.15$ (shown in Fig.~\ref{fig:A1A2}). We assume that
the weights of the edges in both the graphs are all one. We consider
the Markovian temporal network $\mathcal G$ taking values in the set
$\{\mathcal G_1, \mathcal G_2\}$ and having the transition probability
matrix
\begin{equation*}
P = \begin{bmatrix}
0.5419 & 0.4581 \\ 0.1914& 0.8086
\end{bmatrix}. 
\end{equation*}
We fix $\alpha = 0.1321 = 1/(2\rho(\mathcal A))$ and $\zeta_0 =
\onev_2/2$. We sort the nodes in $\mathcal G$ in the decreasing order
of its Katz centrality.

For all $\ell \in \{1, 2\}$ and $\{i, j\} \in \mathcal{E}_\ell$, we
use the increasing cost function $f_{\ell, ij}(\Delta_{\ell, ij}) =
2(1 - (\Delta_{ell, ij}+1)^{-1})$ defined over $[0, 1]$. Let $\bar
\Delta_\ell =  A_\ell$ and $\ubar A_\ell = A_\ell$ for all $\ell$.
Notice that the function~$-f_{\ell, ij}$ is a posynomial in
$\Delta_{\ell, ij}+1 = \Delta_{\ell, ij}+\bar A_{\ell, ij}$ and
therefore satisfies the assumption of
Theorem~\ref{thm:preliminaryProblem}. We let $i$ be the 4th important
node, in terms of the Katz centrality of the original temporal
network~$\mathcal G$. We apply Algorithm~\ref{alg:eigOptimization} and
obtain the optimal additional weights $\Delta_1^\star$ and
$\Delta_2^\star$ with the cost $C = 31.2$. We show the weighted graphs
having the adjacency matrices~$\Delta_1$ and~$\Delta_2$ in
Fig.~\ref{fig:Delta1Delta2}. The Katz centralities of the original and
optimized temporal networks are shown in Fig.~\ref{fig:centralities}.

For another but a larger-scale example, we again consider the
Markovian temporal network whose layers are realizations of the
Erd\"os-R\'enyi graphs. Here we choose the parameters $n=50$, $p =
0.05$, and
\begin{equation*}
P = \begin{bmatrix}
0.1503 & 0.8497 \\ 0.4381& 0.5619
\end{bmatrix}. 
\end{equation*}
We fix $\alpha = 0.1908 = 1/(2\rho(\mathcal A))$ and $\zeta_0 =
\onev_2/2$. We again sort the nodes in $\mathcal G$ in the decreasing
order of its Katz centrality. We set the target node to be $i = 8$,
and use Algorithm~\ref{alg:eigOptimization} to obtain the optimal
additional weights $\Delta_1^\star$ and $\Delta_2^\star$ (illustrated
in Fig.~\ref{fig:Delta1Delta2_50}) with the cost $C = 43.05$. The Katz
centralities of the original and optimized temporal networks are shown
in Fig.~\ref{fig:centralities50}.

\begin{figure}[tb]
\centering
\includegraphics[width=\myfigwidth]{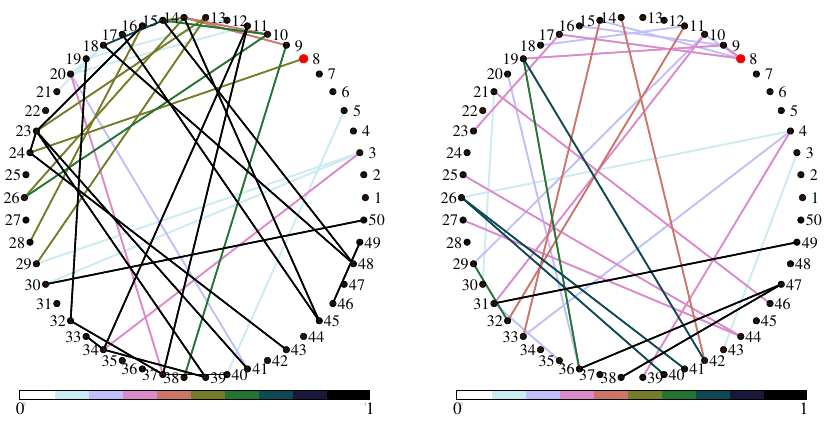}
\caption{The optimal additional weights: $\Delta_1^\star$ (left) and $\Delta_2^\star$ (right)}
\label{fig:Delta1Delta2_50}
\end{figure}

\begin{figure}[tb]
\centering
\includegraphics[width=\myfigwidth]{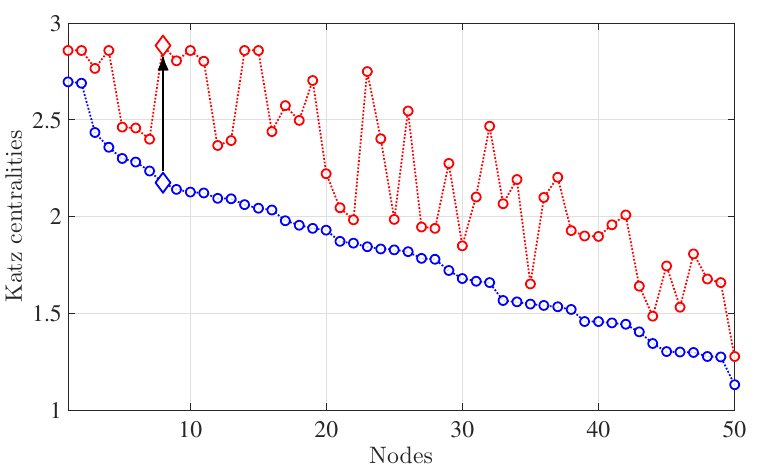}
\caption{Katz centralities: Blue: before optimization. Red: after optimization. Diamonds indicate the target node $i=8$.}
\label{fig:centralities50}
\end{figure}

\section{Conclusion}

In this paper, we have introduced an extension of Katz centrality for
Markovian temporal networks. The definition is based on the solution
of a Markov jump linear system, which is consistent with the standard
definition of Katz centrality for static networks. We have first shown
that the Katz centrality of a Markovian temporal network is the
solution of a linear program, which enables us to efficiently find the
centrality in the case of large network sizes. We have then presented
a convex optimization-based approach for controlling the Katz
centrality by tuning the weights of the temporal network in the most
cost-efficient manner. Numerical simulations have been given to
illustrate the effectiveness of the obtained results.

\end{document}